\documentclass{article}

\usepackage{graphicx}
\usepackage{subfigure} 
\usepackage{color}
\usepackage{algorithm,algorithmic}
\usepackage{amsmath,amsfonts}
\usepackage{array}
\usepackage{verbatim}
\usepackage{graphics}
\usepackage{epstopdf}
\usepackage{multicol}
\usepackage{url}
\usepackage{amsthm}

\newtheorem*{theorem}{Theorem}

\DeclareMathOperator*{\E}{\mathbb{E}}

\begin{document}

\title{Dimension Independent Matrix Square using MapReduce (DIMSUM)}

\date{22 October 2014}

\author{Reza Bosagh Zadeh and Gunnar Carlsson \\
       Stanford University\\
       Stanford CA 94305, USA}
\maketitle

\begin{abstract}
We compute the singular values of an $m \times n$ sparse matrix $A$ in a distributed setting, without communication 
dependence on $m$, which is useful for very large $m$.
In particular, we give a simple nonadaptive sampling scheme where the singular values of $A$ are estimated within relative
error with constant probability. Our proven bounds focus on the MapReduce framework, which has
become the de facto tool for handling such large matrices that cannot be stored or even streamed
through a single machine.

On the way, we give a general method to compute $A^TA$. We preserve singular
values of $A^TA$ with $\epsilon$ relative error with shuffle size $O(n^2/\epsilon^2)$ and reduce-key complexity $O(n/\epsilon^2)$. We further show that if
only specific entries of $A^TA$ are required and $A$ has nonnegative entries, then we can reduce the shuffle
size to $O(n \log(n) / s)$ and reduce-key complexity to $O(\log(n)/s)$, where $s$ is the minimum cosine similarity for the entries being estimated.
All of our bounds are independent of $m$, the larger dimension. We provide open-source implementations in Spark and Scalding, along
with experiments in an industrial setting.
\end{abstract}

\section{Introduction}

There has been a flurry of work to solve problems in numerical linear algebra via fast approximate randomized
algorithms. Starting with \cite{frieze1} many algorithms have been proposed over older algorithms 
\cite{review1,review2,review3,review4,review5,review6,review7,review8,review9,review10,review11,review12,review13}, with
results satisfying the traditional Monte Carlo performance guarantees: small error with high probability.

These proposed algorithms require either streaming, or having access to the entire matrix $A$ on a single
machine, or communicating too much data between machines. 
This is not feasible for very large $m$ (for example $m=10^{13}$). In such cases, $A$
cannot be stored  or streamed through a single machine - let alone be used in computations. For such cases, MapReduce \cite{mapreduce}
has become the de facto tool for handling very large datasets.

MapReduce is a programming model for processing large data sets, typically used to do distributed computing on clusters of commodity computers. With large amount of processing power at hand, it is very tempting to solve problems by brute force. However, we combine clever sampling techniques with the power of MapReduce to extend its utility.

Given an $m \times n$ matrix $A$ with each row having at most $L$ nonzero entries, we show how to compute
the singular values and and right singular vectors of $A$ without dependence on $m$, in a MapReduce environment.
The SVD of $A$ is written $A = U \Sigma V^T$, where $U$ is $m \times n$, $\Sigma $ is $n \times n$, and $V$ is $n\times n$.

We compute $\Sigma$ and $V$. We do this by first computing $A^TA$, which we do without dependence
on $m$. Since $A^TA = V \Sigma^2 V^T$ is $n \times n$, for small $n$ (for example $n=10^4$) we can compute the eigen-decomposition of
$A^TA$ directly and retrieve $V$ and $\Sigma$. What remains is to compute $A^TA$ efficiently and without harming its singular values, which is what the rest of the paper is focused on. 

Our main result is Algorithms \ref{alg:discomapper} and \ref{alg:discoreducer}, 
along with proven guarantees given in Theorem \ref{svdtheorem} which proves a relative error bound using the spectral norm. The
proof uses a new singular value concentration inequality from \cite{latala} that has not  seen much usage
by the theoretical computer science community.

\section{Formal Preliminaries} \label{sec:formal}

Label the columns of $A$ as $c_1, \ldots, c_n$, rows as $r_1, \ldots, r_m$, and the individual entries as $a_{ij}$. The matrix is stored row-by-row on disk and read via mappers. We focus on the case where each dimension is sparse with at most $L$ nonzeros per row  therefore the natural way to store the data is to segment into rows. Throughout the paper we assume the entries of $A$ have been
scaled to be in $[-1, 1]$, which can be done with little communication by finding the largest magnitude element.

We use the matrix spectral norm throughout, which for any $m \times n$ matrix $A$ is defined as
 $$||A||_2 = \max_{x \in \mathbb{R}^m,y \in \mathbb{R}^n} \frac{x^TAy}{||x||_2||y||_2}$$
Unless otherwise denoted, the norm used anywhere in this paper is the spectral norm, which for regular vectors
degenerates to the vector $l_2$ norm.

We concentrate on the regime where $m$ is very large, e.g. $m=10^{13}$, but $n$ is not too large, e.g. $n=10^4$, such that
we can compute the SVD of an $n\times n$ dense matrix on a single machine. The magnitudes of each column is assumed
to be loaded into memory and available to both the mappers and reducers. The
magnitudes of each column are natural values to have computed already,
or can be computed with a trivial mapreduce.

\subsection{Naive Computation}

The naive way to compute $A^TA$ on MapReduce is to materialize all dot products between columns of $A$ trivially.
For purposes of demonstrating the complexity measures
for MapReduce, we briefly write down the Naive algorithm to compute $A^TA$.

\algsetup{indent=2em}
\begin{algorithm}[H]
\caption{NaiveMapper$(r_i)$} \label{alg:naivemapper}
\begin{algorithmic} [4]
\FOR{all pairs $(a_{ij}, a_{ik})$ in $r_i$}
	\STATE  Emit $((c_j, c_k) \rightarrow a_{ij} a_{ik})$
\ENDFOR
\end{algorithmic}
\end{algorithm}

\algsetup{indent=2em}
\begin{algorithm}[H]
\caption{NaiveReducer$((c_i, c_j), \langle v_1, \ldots, v_R \rangle)$} \label{alg:naivereducer}
\begin{algorithmic} [4]
\STATE output $ c_i^T c_j \rightarrow \sum_{i=1}^{R} v_i$
\end{algorithmic}
\end{algorithm}

\subsection{Complexity Measures}

There are two main complexity measures for MapReduce: ``shuffle size", and ``reduce-key complexity". 
These complexity measures together capture the bottlenecks when handling data on multiple machines: first we can't have
too much communication between machines, and second we can't overload a single machine. 
The number of emissions in the map phase is called the ``shuffle size", since that data needs to be shuffled around the network to reach the correct reducer. The maximum number of items reduced to a single key is called the ``reduce-key complexity" and measures
how overloaded a single machine may become \cite{ashishperf}.

It can be easily seen that the naive approach for computing $A^TA$ will have $O(mL^2)$ emissions, which for the example parameters we gave ($m=10^{13}, n=10^4, L=20)$ is infeasible.  Furthermore, the maximum number of items reduced to a single key can be as large as $m$. Thus the ``reduce-key complexity" for the naive scheme is $m$.

We can drastically reduce the shuffle size and reduce-key complexity by some clever sampling with the DIMSUM scheme
described in this paper. In this case, the output of the reducers are random variables whose expectations are cosine similarities i.e. normalized entries of $A^TA$. Two proofs are needed to justify the effectiveness of this scheme. First, that the expectations are indeed correct and obtained with constant probability, and second, that the shuffle size is greatly reduced. We prove both of these claims. In particular, in addition to correctness, we prove that for relative error $\epsilon$, the shuffle size of our scheme is only $O(n^2/\epsilon^2)$, with no dependence on the dimension $m$, hence the title of this paper. 

This means as long as there are enough mappers to read the data, our sampling scheme can be used to make the shuffle size tractable. Furthermore, each reduce-key gets at most $O(n/\epsilon^2)$ values, thus making the reduce-key complexity tractable, too. Within Twitter Inc, we use the DIMSUM sampling scheme to compute similar users \cite{blogpost,wtfpaper}. We have also used the scheme to find highly similar pairs of words, by taking each dimension to be the indicator vector that signals in which tweets the word appears. We empirically verified the proven claims in this paper, but do not report experimental results since we are primarily focused on the proofs.

\subsection{Related Work}

\cite{frieze1} introduced a sampling procedure where rows and columns of $A$ are picked with probabilities
proportional to their squared lengths and used that to compute an approximation to $A^TA$. 
Later \cite{achlioptas} and \cite{arora} improved the sampling procedure.
To implement these approximations to $A^TA$ on MapReduce one would
need a shuffle size dependent on $m$ or overload a single machine. 
We improve this to be independent of $m$ both in shuffle size and reduce-key complexity.

Later on \cite{adaptive} found an adaptive sampling scheme to improve
the scheme of \cite{frieze1}. Since the scheme is adaptive, it would require
too much communication between machines holding $A$. In particular a MapReduce
implementation would still have shuffle size dependent on $m$, 
and require many (more than 1) iterations.

There has been some effort to reduce the number of passes required through
the matrix $A$ using little memory, in the streaming model \cite{streaminglinear}. The
question was posed by \cite{muthukrishnan} to determine in the streaming model
various linear algebraic quantities. The problem was posed again by \cite{sarlos} who asked
about the time and space required for an algorithm not using too many passes. The streaming
model is a good one if all the data can be streamed through a single machine, but with $m$
so large, it is not possible to stream $A$ through a single machine. Splitting the work of reading
$A$ across many mappers is the job of the MapReduce implementation and one of its major
advantages \cite{mapreduce}.

There has been recent work specifically targeted at computing the SVD on MapReduce \cite{austin} 
in a stable manner via $QR$ factorizations and bypassing $A^TA$, with shuffle size and reduce-key complexity
 both dependent on $m$.

In addition to computing entries of $A^TA$, our sampling scheme can be used to implement many similarity measures.
We can use the scheme to efficiently compute four similarity measures: Cosine, Dice, Overlap, and the Jaccard similarity measures, 
with details and experiments given in \cite{disco,wtfpaper}, whereas this paper is more focused on matrix computations
and an open-source implementation.

\section{Algorithm}

Our algorithm to compute $A^TA$ efficiently is given below in Algorithms \ref{alg:discomapper} and \ref{alg:discoreducer}.

\algsetup{indent=2em}
\begin{algorithm}[H]
\caption{DIMSUMMapper$(r_i)$}\label{alg:discomapper}
\begin{algorithmic} [4]
\FOR{all pairs $(a_{ij}, a_{ik})$ in $r_i$}
	\STATE With probability $$\min\left(1,\gamma \frac{1}{||c_j|| ||c_k||}\right)$$ emit $((c_j, c_k) \rightarrow a_{ij} a_{ik})$
\ENDFOR
\end{algorithmic}
\end{algorithm}

\algsetup{indent=2em}
\begin{algorithm}[H]
\caption{DIMSUMReducer$((c_i, c_j), \langle v_1, \ldots, v_R \rangle)$}\label{alg:discoreducer}
\begin{algorithmic} [8]
\IF {$ \frac{\gamma}{||c_i|| ||c_j||} > 1$}
	\STATE output $b_{ij} \rightarrow \frac{1}{||c_i|| ||c_j||} \sum_{i=1}^{R} v_i$
\ELSE 
	\STATE output $b_{ij} \rightarrow \frac{1}{\gamma} \sum_{i=1}^{R} v_i$
\ENDIF

\end{algorithmic}
\end{algorithm}

It is important to observe what happens if the output `probability' is greater than 1. We certainly Emit, but when the output probability is greater than 1, care must be taken while reducing to scale by the correct factor,
 since it won't be correct to divide by $\gamma$, which is
the usual case when the output probability is less than 1. Instead, the sum in Algorithm \ref{alg:discoreducer} obtains the dot product, because for the pairs where the output probability is greater than 1, DIMSUMMapper effectively always emits. 
We do not repeat this point later in the paper, nonetheless it is an important one which arises during implementation.

\section{Correctness}

Before we move onto the correctness of the algorithm, we must state Latala's Theorem \cite{latala}.
This theorem talks about a general model of random matrices whose entries are independent
centered random variables with some general distribution (not necessarily normal). The
largest singular value (the spectral norm) can be estimated by Latala's theorem for
general random matrices with non-identically distributed entries:

\begin{theorem} \label{latala} \cite{latala}. Let $X$ be a random matrix whose entries $x_{ij}$
are independent centered random variables with finite fourth moment. Denoting $||X||_2$ as the
matrix spectral norm, we have

$$\E ||X||_2  \leq C \left[  \max_i \left(\sum_j \E x_{ij}^2 \right)^{1/2} + \max_j \left(\sum_i \E x_{ij}^2 \right)^{1/2}
 + \left(\sum_{i,j} \E x_{ij}^4 \right)^{1/4}   \right].$$
\end{theorem}

We analyze the second and fourth central moments of the entries of the estimate for $A^TA$,
and show that by Latala's theorem, the singular values are preserved with constant probability.
Let the matrix output by the DIMSUM algorithm be called $B$ with entries $b_{ij}$. Notice that
this is an $n \times n$ matrix of cosine similarities between columns of $A$. 
Define a diagonal matrix $D$ with $d_{ii} = ||c_i||$.
Then we can
undo the cosine similarity normalization to obtain an estimate for $A^TA$ by using $DBD$. This
effectively uses the cosine similarities between columns of $A$ as an importance sampling scheme.
We have the following theorem:

\begin{theorem} \label{svdtheorem} Let $A$ be an $m \times n$ matrix with $m > n$.
If $\gamma = \Omega(n/\epsilon^2)$ and $D$ a diagonal matrix with entries $d_{ii} = ||c_i||$, then the matrix $B$ output by DIMSUM (Algorithms \ref{alg:discomapper} and \ref{alg:discoreducer}) satisfies,
$$\frac{||DBD - A^TA||_2}{||A^TA||_2} \leq \epsilon $$
 with probability at least $1/2$.
 \end{theorem}
\begin{proof}

We define the indicator variable $X_{ijk}$ to take value $a_{ki} a_{kj}$ with probability $p_{ij} = \gamma \frac{1}{||c_i|| ||c_j||}$
 on the $k$'th call to DIMSUMMapper, and zero with probability $1-p_{ij}$. 
$$
X_{ijk} =
\left\{
	\begin{array}{ll}
		a_{kj} a_{kj}  & \mbox{with prob. }  p_{ij} \\
		0 & \mbox{with prob. }  1- p_{ij}
	\end{array}
\right.
 $$
Then we can write the entries of $B$ as 
 $$b_{ij} =   \frac{1}{\gamma} \sum_{k=1}^m X_{ijk}$$

Since we give relative error bounds and singular values scale trivially,
we can assume $A$ has all entries in $[0,1]$. 
i.e. any scaling of the input matrix will have the same relative error guarantee.
This assumption will be useful because we first prove an absolute error bound, then use
that to prove a relative error bound.
It should be clear from the definitions that in expectation
$$E[B] = D^{-1} A^TA D^{-1} \text{ and } E[DBD] = A^TA$$
With these definitions, 
we now move onto bounding $\E[ ||B - D^{-1} A^TA D^{-1}||]$. With the goal of invoking Latala's theorem,
we analyze $\E [(b_{ij} - E b_{ij})^2]$ and $\E [(b_{ij} - E b_{ij})^4]$.

 Now define $\#(i,j)$ as the number
 of dimensions in which $c_i$ and $c_j$ are \textit{both nonzero}, i.e. the number of $k$ for which $a_{ki} a_{kj}$ is nonzero, and further define
 $i \cap j$ as the set of indices for which $a_{ki} a_{kj}$ is nonzero.

Clearly, $\E [(b_{ij} - E b_{ij})^2]$ is the variance of $b_{ij}$, which is the sum of $\#(i,j)$ weighted indicator random variables. 
Thus we have 

$$\E [(b_{ij} - E b_{ij})^2] = \text{Var}[b_{ij}] = \frac{1}{\gamma^2} \sum_{k \in i \cap j} \text{Var}[X_{ijk}] $$
$$ = \frac{1}{\gamma^2} \sum_{k \in i \cap j} a_{ki}^2 a_{kj}^2 p_{ij}(1-p_{ij})$$
$$ \leq  \frac{1}{\gamma^2} \sum_{k \in i \cap j} a_{ki}^2 a_{kj}^2 p_{ij} $$
$$= \frac{1}{\gamma^2} \sum_{k \in i \cap j} a_{ki}^2 a_{kj}^2  \gamma \frac{1}{||c_i|| ||c_j||} $$

Now by the Arithmetic-Mean Geometric-Mean inequality,
$$ \leq \frac{1}{2 \gamma^2} \sum_{k \in i \cap j} a_{ki}^2 a_{kj}^2  \gamma \left(\frac{1}{||c_i||^2} + \frac{1}{||c_j||^2}\right) $$
$$ = \frac{1}{2 \gamma} \sum_{k \in i \cap j} a_{ki}^2 a_{kj}^2  \left(\frac{1}{||c_i||^2} + \frac{1}{||c_j||^2}\right) $$
$$ \leq \frac{1}{\gamma} \sum_{k \in i \cap j} a_{ki}^2 a_{kj}^2  \left(\frac{1}{||c_j||^2}\right) $$
$$ \leq \frac{1}{\gamma} \sum_{k \in i \cap j}    \frac{a_{kj}^2}{||c_j||^2} \leq \frac{1}{\gamma} $$

Thus we have $E [(b_{ij} - E b_{ij})^2] \leq \frac{1}{\gamma}$. It remains to bound the fourth central moment of $b_{ij}$. We use a  counting trick to achieve this bound:

$$\E [(b_{ij} - E b_{ij})^4] = \frac{1}{\gamma^4} \E \left[ \left(\sum_{k \in i \cap j} X_{ijk} - a_{ki} a_{kj}  p_{ij}  \right)^4\right]$$
$$= \frac{1}{\gamma^4} \E \left[ \sum_{q,r,s,t \in i \cap j} (X_{ijq} - a_{qi} a_{qj}  p_{ij})(X_{ijr} - a_{ri} a_{rj}  p_{ij})(X_{ijs} - a_{si} a_{sj}  p_{ij})(X_{ijt} - a_{ti} a_{tj}  p_{ij}) \right]$$
$$= \frac{1}{\gamma^4}  \sum_{q,r,s,t \in i \cap j} \E \left[ (X_{ijq} - a_{qi} a_{qj}  p_{ij})(X_{ijr} - a_{ri} a_{rj}  p_{ij})(X_{ijs} - a_{si} a_{sj}  p_{ij})(X_{ijt} - a_{ti} a_{tj}  p_{ij}) \right]$$
which effectively turns this into a counting problem. The terms in the sum on the last expression are 0 unless either $q=r=s=t$, which happens $\#(i,j)$ times, or there are two pairs of matching indices, which happens  ${\#(i,j) \choose 2} {4 \choose 2}$ times. Continuing, this gives us
$$= 	 \frac{1}{\gamma^4}\sum_{k \in i \cap j} \E[(X_{ijk} -  a_{ki} a_{kj}  p_{ij})^4]  +     \frac{1}{\gamma^4}\sum_{q,r \in i \cap j} \text{Var}[X_{ijq}] \text{Var}[X_{ijr}] $$
$$=  \frac{1}{\gamma^4}	\sum_{k \in i \cap j}  a_{ki}^4 a_{kj}^4  [p_{ij}^4(1-p_{ij}) + (1-p_{ij})^4p_{ij}]  $$
$$+    \frac{1}{\gamma^4} \sum_{q,r \in i \cap j} a_{qi}^2 a_{qj}^2 p_{ij}(1-p_{ij}) a_{ri}^2 a_{rj}^2 p_{ij}(1-p_{ij})$$
$$ \leq  \frac{1}{\gamma^4}	\sum_{k \in i \cap j}  a_{ki}^4 a_{kj}^4  p_{ij} +    \frac{1}{\gamma^4} \sum_{q,r \in i \cap j} a_{qi}^2 a_{qj}^2  a_{ri}^2 a_{rj}^2 p_{ij}^2$$
$$ =  \frac{1}{\gamma^3}  \frac{1}{||c_i|| ||c_j||}\sum_{k \in i \cap j}  a_{ki}^4 a_{kj}^4  +    \frac{1}{\gamma^2}   \frac{1}{||c_i||^2 ||c_j||^2} \sum_{q,r \in i \cap j} a_{qi}^2 a_{qj}^2  a_{ri}^2 a_{rj}^2$$
by the Arithmetic-Mean Geometric-Mean inequality,
$$ \leq  \frac{1}{2\gamma^3}  (\frac{1}{||c_i||^2}+ \frac{1}{||c_j||^2})\sum_{k \in i \cap j}  a_{ki}^4 a_{kj}^4  +    \frac{1}{\gamma^2}   \frac{1}{||c_i||^2 ||c_j||^2} \sum_{q,r \in i \cap j} a_{qi}^2 a_{qj}^2  a_{ri}^2 a_{rj}^2$$
and since entries $a_{ij} \in [0,1]$,
$$ \leq  \frac{1}{2\gamma^3}  (\frac{1}{||c_i||^2}+ \frac{1}{||c_j||^2})\sum_{k \in i \cap j}  a_{ki}^2 a_{kj}^2  +    \frac{1}{\gamma^2}   \frac{1}{||c_i||^2 ||c_j||^2} \sum_{q,r \in i \cap j} a_{qi}^2 a_{rj}^2$$
$$ \leq  \frac{1}{\gamma^3}  \frac{1}{||c_i||^2}\sum_{k \in i \cap j}  a_{ki}^2  +    \frac{1}{\gamma^2}   \frac{1}{||c_i||^2 ||c_j||^2} \sum_{q,r \in i \cap j} a_{qi}^2 a_{rj}^2$$
$$ \leq  \frac{1}{\gamma^3} +    \frac{1}{\gamma^2}   $$
for $\gamma \geq 1$,
$$ \leq  \frac{2}{\gamma^2}  $$

Thus we have that $\E [(b_{ij} - E b_{ij})^4] \leq \frac{2}{\gamma^2}$, and from the above we have $\E [(b_{ij} - E b_{ij})^2] \leq \frac{1}{\gamma}$. Plugging
these into Theorem \ref{latala}, we can bound the \textit{absolute error} between $B$ and $D^{-1} A^TA D^{-1}$,

$$\E[ ||B - D^{-1} A^TA D^{-1}||] \leq C_0 [  \max_i \left(\sum_j \E [(b_{ij} - E b_{ij})^2] \right)^{1/2} $$
$$ + \max_j \left(\sum_i \E [(b_{ij} - E b_{ij})^2]  \right)^{1/2}
      + \left(\sum_{i,j} \E [(b_{ij} - E b_{ij})^4]  \right)^{1/4}  ]$$
$$\leq C_0 [ \left( \frac{n}{\gamma} \right)^{1/2} + \left( \frac{n}{\gamma} \right)^{1/2}
      + \left( \frac{2n^2}{\gamma^2}  \right)^{1/4}  ]$$
$$\leq C_1 \left( \frac{n}{\gamma} \right)^{1/2} $$
where $C_0$ and $C_1$ are absolute constants. Thus we have that $$\E[ ||B - D^{-1} A^TA D^{-1}||] \leq  C_1 \left( \frac{n}{\gamma} \right)^{1/2}$$

Setting $\gamma = 4 C_1^2 \frac{n}{\epsilon^2}$, gives $$E[ ||B - D^{-1} A^TA D^{-1}||] \leq \epsilon/2$$
Thus by the Markov inequality we have with probability at least $1/2$, 
$$ ||B - D^{-1} A^TA D^{-1}|| \leq \epsilon$$

Which gives us an absolute error bound between $B$ and $D^{-1} A^TA D^{-1}$. It remains to get a relative error bound between
$DBD$ and  $A^TA$,

$$\frac{||DBD - A^TA||}{||A^TA||} =  \frac{||D(B - D^{-1} A^TA D^{-1})D||}{||A^TA||}$$
by the submultiplicative property of the spectral norm,
$$\leq \frac{||D||^2 ||B - D^{-1} A^TA D^{-1}||}{||A^TA||}$$
Now since $D$ is a diagonal matrix with positive entries, its spectral norm is its largest entry, i.e. the largest column magnitude, call it $c_*$,
$$\leq \frac{ c_*^2 ||B - D^{-1} A^TA D^{-1}||}{||A^TA||}$$
Now we use another property of the spectral norm to lowerbound $||A^TA||$,
 $$||A^TA|| = \max_{x,y \in \mathbb{R}^n} \frac{x^TA^TAy}{||x||||y||}$$
 Setting $x,y$ to be indicator vectors to pick out the $i$'th diagonal entry of $A^TA$, we have that $||A^TA|| \geq c_*^2$ since $c_*^2$ is
 some entry in the diagonal of $A^TA$. In addition to allowing us to bound the fourth central moment, this is yet
 another reason why we picked the sampling probabilities in Algorithm \ref{alg:discomapper}.
 Finally, continuing from above armed with this lower bound,

$$\frac{||DBD - A^TA||}{||A^TA||} \leq \frac{ c_*^2 ||B - D^{-1} A^TA D^{-1}||}{||A^TA||}$$
$$\leq \frac{ c_*^2 \epsilon}{||A^TA||}$$
$$\leq \frac{ c_*^2 \epsilon}{c_*^2}$$
$$= \epsilon$$
with probability at least $1/2$.

\end{proof}

Although we had to set $\gamma = \Omega(n/\epsilon^2)$ to estimate
the singular values, if instead of the singular values
we are interested in individual entries of $A^TA$ that are large, we can get away setting
$\gamma$ significantly smaller, and thus reducing shuffle size. In particular if two columns have
high cosine similarity, we can estimate the corresponding entry in $A^TA$ with much less
computation. Here we define cosine
similarity as the normalized dot product $$\cos(c_i, c_j) = \frac{c_i^Tc_j}{||c_i||||c_j||}$$

\begin{theorem} \label{cosinecorrect}
Let $A$ be an $m\times n$ matrix with entries in $[0,1]$. For any two columns $c_i$ and $c_j$ having $\cos(c_i,c_j) \geq \epsilon$, let 
 $B$ be the output of DIMSUM with entries $b_{ij} =   \frac{1}{\gamma} \sum_{k=1}^m X_{ijk}$ with 
 $X_{ijk}$ as defined in Theorem \ref{svdtheorem}.
Now if $\gamma \geq \alpha/\epsilon$, then we have,
$$\Pr \left[ ||c_i||||c_j|| b_{ij} > (1+\delta)[A^TA]_{ij} \right] \leq \left(\frac{e^\delta}{(1+\delta)^{(1+\delta)}}\right)^\alpha $$
and 
$$\Pr \left[ ||c_i||||c_j|| b_{i,j} < (1-\delta)[A^TA]_{ij} \right]  < \exp(-\alpha\delta^2/2)$$
\end{theorem}
\begin{proof}
We use $||c_i||||c_j|| b_{i,j}$ as the estimator for $[A^TA]_{ij}$. Note that 
$$\mu_{ij} = \E [\sum_{k=1}^m X_{ijk}] =   \gamma \frac{c_i^T c_j}{||c_i|| ||c_j||} =  \gamma \cos(x,y) \geq \alpha$$ 
Thus by the multiplicative form of the Chernoff bound,
$$\Pr \left[ ||c_i||||c_j|| b_{ij} > (1+\delta)[A^TA]_{ij}  \right]
=\Pr \left[ \gamma \frac{||c_i||||c_j||}{||c_i||||c_j||} b_{ij} > \gamma (1+\delta)\frac{[A^TA]_{ij}}{||c_i||||c_j||}  \right]$$
$$=\Pr \left[ \sum_{k=1}^m X_{ijk} > (1+\delta) \E[\sum_{k=1}^m X_{ijk}]  \right]
\leq \left(\frac{e^\delta}{(1+\delta)^{(1+\delta)}}\right)^\alpha$$

Similarly, by the other side of the multiplicative Chernoff bound, we have
$$\Pr \left[ ||c_i||||c_j|| b_{ij} < (1+\delta)[A^TA]_{ij}  \right]
=\Pr \left[ \gamma \frac{||c_i||||c_j||}{||c_i||||c_j||} b_{ij} < \gamma (1+\delta)\frac{[A^TA]_{ij}}{||c_i||||c_j||}  \right]$$
$$=\Pr \left[ \sum_{k=1}^m X_{ijk} < (1+\delta) \E[\sum_{k=1}^m X_{ijk}]  \right]$$
$$ < \exp(-\mu_{ij}\delta^2/2) \leq \exp(-\alpha \delta^2/2)$$
\end{proof}

\section{Shuffle Size}

Define $H$ as the smallest nonzero entry of $A$ in magnitude, after the entries of $A$ have been scaled
to be in $[0,1]$. For example when $A$ has entries in $\{0,1\}$, we have $H=1$. 

\begin{theorem} \label{smallshuffle}
Let $A$ be an $m \times n$ sparse matrix with
at most $L$ nonzeros per row. The expected shuffle size for DIMSUMMapper is $O(n L \gamma / H^2)$.
\end{theorem}
\begin{proof}

Define $\#(c_i,c_j)$ as the number
 of dimensions in which $c_i$ and $c_j$ are \textit{both nonzero}, i.e. number of $k$ for which $a_{ki} a_{kj}$ is nonzero. 

The expected contribution from each pair of columns will constitute the shuffle size:
$$\sum_{i=1}^{n} \sum_{j=i+1}^n \sum_{k=1}^{\#(c_i, c_j)} \text{Pr}[\text{DIMSUMMapper}(c_i, c_j)] $$
$$= \sum_{i=1}^{n} \sum_{j=i+1}^n \#(c_i, c_j) \text{Pr}[\text{DIMSUMMapper}(c_i, c_j)]  $$
$$= \sum_{i=1}^{n} \sum_{j=i+1}^n \gamma \frac{\#(c_i, c_j)}{||c_i|| ||c_j||} $$
By the Arithmetic-Mean Geometric-Mean inequality,
$$ \leq \frac{\gamma}{2} \sum_{i=1}^{n} \sum_{j=i+1}^n \#(c_i, c_j)( \frac{1}{||c_i||^2} + \frac{1}{||c_j||^2}) $$
$$\leq \gamma \sum_{i=1}^{n} \frac{1}{||c_i||^2} \sum_{j=1}^n \#(c_i, c_j)$$
$$\leq \gamma \sum_{i=1}^{n} \frac{1}{||c_i||^2} L ||c_i||^2 / H^2 = \gamma Ln/H^2 $$

The first inequality holds because of the Arithmetic-Mean Geometric-Mean inequality applied to $\{1/||c_i||, 1/||c_j||\}$. The
last inequality holds because $c_i$ can co-occur with at most $||c_i||^2 L / H^2$ other columns. 
It is easy to see via Chernoff bounds that the above shuffle size is obtained with high probability.

\end{proof}

\begin{theorem} \label{smallshuffle}
Let $A$ be an $m \times n$ sparse matrix $A$ with at most $L$ nonzeros per row.
The  shuffle size for any algorithm computing those entries of $A^TA$ for which $\cos(i,j)\geq \epsilon$
 is at least $\Omega(n L)$.
\end{theorem}
\begin{proof}

To see the lowerbound, we construct a dataset consisting of $n/L$ distinct rows of length $L$,
furthermore each row is duplicated $L$ times. To construct this dataset, consider 
grouping the columns into $n/L$ groups, each group containing $L$ columns. A row is associated
with every group, consisting of all the columns in the group. This row is then repeated $L$ times.
In each group, it is trivial to check that all pairs of columns have cosine similarity exactly 1. There are ${L \choose 2}$
pairs for each group and there are $n/L$ groups, making for a total of
$(n/L) {L \choose 2} = \Omega(nL)$ pairs with similarity 1, and thus also at least $\epsilon$.
Since any algorithm that purports to accurately calculate highly-similar pairs must at least \textit{output} them,
and there are $\Omega(nL)$ such pairs, we have the lower bound.
\end{proof}

\begin{theorem} \label{smallreducekey}
Let $A$ be an $m \times n$ matrix with non-negative entries.
The expected number of values mapped to a single key  by DIMSUMMapper is at most $\gamma / H^2$.
\end{theorem}
\begin{proof}
Note that the output of DIMSUMReducer is a number between 0 and 1. Since this is obtained
by normalizing the sum of all values reduced to the key by at most $\gamma$, and all summands are at least
$H^2$, we get that the number of summands is at most $\gamma/H^2$.
\end{proof}

\section{Reducing Computation}
In DIMSUMMapper, it is required to generate $L \choose 2$ random numbers for each row, which
doesn't cause communication between machines, but does require computation. We can reduce this
computation by moving the random number generation as in Algorithm \ref{alg:v2mapper}, which uses the summation reducer.
However, 
in Algorithm \ref{alg:v2mapper} it is no longer true that the $b_{ij}$ are pairwise independent, and thus
an analog of Theorem \ref{svdtheorem} does not hold for Algorithm \ref{alg:v2mapper}. 
However, Theorem \ref{cosinecorrect} does hold, as it does
not require the $b_{ij}$ to be be pairwise independent, and 
so when cosine similarities are sought, this is a useful modification.

\algsetup{indent=2em}
\begin{algorithm}[H]
\caption{LeanDIMSUMMapper$(r_i)$}\label{alg:v2mapper}
\begin{algorithmic} [40]
\FOR{all  $a_{ij}$ in $r_i$}
	\STATE With probability $\min\left(1, \frac{\sqrt{\gamma}}{||c_j||}\right)$\\
\FOR{all  $a_{ik}$ in $r_i$}
	\STATE With probability $\min\left(1, \frac{\sqrt{\gamma}}{ ||c_k||}\right)$\\
	emit $(b_{jk} \rightarrow \frac{a_{ij} a_{ik}}{ \min(\sqrt{\gamma}, ||c_j||) \min(\sqrt{\gamma}, ||c_k||)} )$
\ENDFOR
\ENDFOR
\end{algorithmic}
\end{algorithm}

\section{Experiments and Open Source Code}
We run DIMSUM daily on a production-scale ads dataset at Twitter \cite{blogpostdimsum}.
Upon replacing the traditional cosine similarity computation in late June 2014, 
we observed 40\% improvement in several performance measures, plotted in Figure \ref{exps}.
The y-axis ranges from 0 to hundred of terabytes, where the exact amount is kept confidential.

\begin{figure}[h]
\caption{\label{exps} DIMSUM turned on in late June. The y-axis ranges from 0 bytes to hundreds of terabytes.}
\centering
\includegraphics[width=10cm]{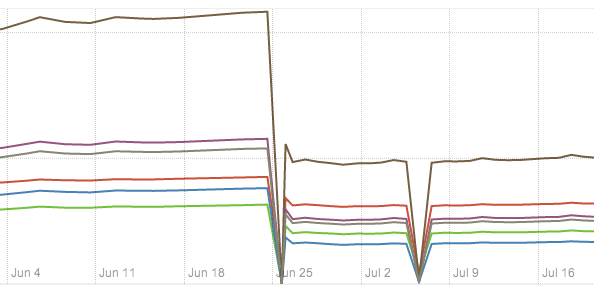} \includegraphics[width=4cm]{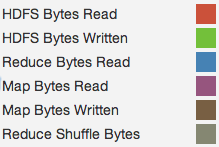}
\end{figure}

We have contributed an implementation of DIMSUM to two open source projects: Scalding and Spark \cite{sparkpaper}.
The Spark implementation is widely distributed by many commercial vendors that
package Spark with their industrial cluster installations.
\begin{itemize}
\item Spark github pull-request: \url{https://github.com/apache/spark/pull/1778}
\item Scalding github pull-request: \url{https://github.com/twitter/scalding/pull/833}
\end{itemize}

\section{Conclusions and Future Directions}

We presented the DIMSUM algorithm to compute
$A^TA$ for an $m \times n$ matrix $A$ with $m > n$.
All of our results are provably independent of the dimension $m$, meaning
that apart from the initial cost of trivially reading in the data, all subsequent
operations are independent of the dimension, the dimension can thus be very large.

Although we used $A^TA$ in the context of computing singular values, there are likely
other linear algebraic quantities that can benefit from having a provably efficient and accurate
MapReduce implementation of $A^TA$. For example if
one wishes to use the estimate for $A^TA$  in solving 
the normal equations in the ubiquitous least-squares problem
$$A^TA x = A^Ty$$
 then the guarantee given by Theorem
\ref{svdtheorem} gives some handle on the problem, although a concrete
error bound is left for future work.

\section{Acknowledgements}

We thank the Twitter Personalization and Recommender systems team for allowing
us to use production data from the live Twitter site for experiments (not reported), and Kevin Lin 
for the implementation in the Twitter Ads team.
We also thank Jason Lee, Yuekai Sun, and Ernest Ryu
from Stanford ICME for  valuable discussions. 
Finally we thank the Stanford student group: Computational Consulting
and all its members for their help. 
\newpage
\bibliographystyle{plain}
\bibliography{dimsum}

\end{document}